\documentclass[submission]{FPSAC2024}
\usepackage[utf8]{inputenc}
\usepackage{tikz}
\usepackage{graphicx}
\usepackage{microtype}

\newcommand{\tikzrect}[1]{}
\newcommand{\tikzgener}[1]{}
\tikzstyle{wall}=[line width=1pt]

\definecolor{p1col1}{rgb}{0.266122, 0.486664, 0.802529}
\definecolor{p1col2}{rgb}{0.513417, 0.72992, 0.440682}
\definecolor{p1col3}{rgb}{0.863512, 0.670771, 0.236564}
\definecolor{p1col4}{rgb}{0.857359, 0.131106, 0.132128}
\definecolor{p1col5}{rgb}{0.52734, 0.24218, 0.24218}
\definecolor{p1col6}{rgb}{0.45, 0.45, 0.45}

\definecolor{p2col1}{rgb}{0.264135, 0.201429, 0.745889}
\definecolor{p2col2}{rgb}{0.256326, 0.430921, 0.808553}
\definecolor{p2col3}{rgb}{0.324106, 0.60897, 0.708341}
\definecolor{p2col4}{rgb}{0.439128, 0.704968, 0.52925}
\definecolor{p2col5}{rgb}{0.597181, 0.742185, 0.36771}
\definecolor{p2col6}{rgb}{0.764712, 0.728302, 0.273608}
\definecolor{p2col7}{rgb}{0.88018, 0.631684, 0.227665}
\definecolor{p2col8}{rgb}{0.897354, 0.41824, 0.185007}
\definecolor{p2col9}{rgb}{0.857359, 0.131106, 0.132128}

\newcommand{\pone}{
\begin{tikzpicture}[baseline=1.2mm,scale=0.1275]
\protect\draw[wall] (1.2,0)--(1.2,2.8);
\protect\draw[wall] (4,1.2)--(1.2,1.2);
\protect\draw[wall] (2.8,4)--(2.8,1.2);
\protect\draw[wall] (0,2.8)--(2.8,2.8);
\end{tikzpicture}
}

\newcommand{\ptwo}{
\begin{tikzpicture}[baseline=1.2mm,scale=0.1275]
\protect\draw[wall] (2.8,0)--(2.8,2.8);
\protect\draw[wall] (4,2.8)--(1.2,2.8);
\protect\draw[wall] (1.2,4)--(1.2,1.2);
\protect\draw[wall] (0,1.2)--(2.8,1.2);
\end{tikzpicture}
}

\newcommand{\pthr}{
\begin{tikzpicture}[baseline=1.2mm,scale=0.1275]
\protect\draw[wall] (2,0)--(2,4);
\protect\draw[wall] (0,1.2)--(2,1.2);
\protect\draw[wall] (2,2.8)--(4,2.8);
\end{tikzpicture}
}

\newcommand{\pfiv}{
\begin{tikzpicture}[baseline=1.2mm,scale=0.1275]
\protect\draw[wall] (2,0)--(2,4);
\protect\draw[wall] (0,2.8)--(2,2.8);
\protect\draw[wall] (2,1.2)--(4,1.2);
\end{tikzpicture}
}

\newcommand{\pfou}{
\begin{tikzpicture}[baseline=1.2mm,scale=0.1275]
\protect\draw[wall] (0,2)--(4,2);
\protect\draw[wall] (2.8,2)--(2.8,4);
\protect\draw[wall] (1.2,0)--(1.2,2);
\end{tikzpicture}
}

\newcommand{\psix}{
\begin{tikzpicture}[baseline=1.2mm,scale=0.1275]
\protect\draw[wall] (0,2)--(4,2);
\protect\draw[wall] (2.8,0)--(2.8,2);
\protect\draw[wall] (1.2,2)--(1.2,4);
\end{tikzpicture}
}

\newcommand{\psev}{
\begin{tikzpicture}[baseline=1.2mm,scale=0.1275]
\protect\draw[wall] (0,2)--(1.65,2);
\protect\draw[wall] (1.65,0)--(1.65,4);
\protect\draw[wall] (2.35,0)--(2.35,4);
\protect\draw[wall] (2.35,2)--(4,2);
\end{tikzpicture}
}

\newcommand{\peig}{
\begin{tikzpicture}[baseline=1.2mm,scale=0.1275]
\protect\draw[wall] (2,0)--(2,1.65);
\protect\draw[wall] (0,1.65)--(4,1.65);
\protect\draw[wall] (0,2.35)--(4,2.35);
\protect\draw[wall] (2,2.35)--(2,4);
\end{tikzpicture}
}

\def\P#1{\ifnum#1=1 $\pone$
\else\ifnum#1=2 $\ptwo$
\else\ifnum#1=3 $\pthr$
\else\ifnum#1=4 $\pfou$
\else\ifnum#1=5 $\pfiv$
\else\ifnum#1=6 $\psix$
\else\ifnum#1=7 $\psev$
\else\ifnum#1=8 $\peig$
 \else            \relax
\fi\fi\fi\fi\fi\fi\fi\fi}

\newcommand{\td}{
\begin{tikzpicture}[baseline=0.5mm,scale=0.17]
\protect\draw[wall] (0,1.7)--(2,1.7);
\protect\draw[wall] (1,0)--(1,1.7);
\end{tikzpicture}
}

\newcommand{\tu}{
\begin{tikzpicture}[baseline=0.5mm,scale=0.17]
\protect\draw[wall] (0,0)--(2,0);
\protect\draw[wall] (1,0)--(1,1.7);
\end{tikzpicture}
}

\def\internallinenumbers{}
\definecolor{darkgreen}{rgb}{0,0.4,0}
\definecolor{BrickRed}{rgb}{0.65,0.08,0}
\hypersetup{colorlinks=true,linkcolor=blue,citecolor=darkgreen,filecolor=BrickRed,urlcolor=darkgreen}
\newtheorem{thm}{Theorem}
\newtheorem{lem}[thm]{Lemma}
\newcommand{\oeis}[1]{\text{\href{https://oeis.org/#1}{{\small \tt #1}}}}

\newcommand{\N}{\mathbb{N}}

\newcommand{\RR}{\mathcal{R}}
\newcommand{\V}{\mathcal{V}}

\usepackage{mathabx}

\title[From geometry to generating functions: rectangulations and permutations]{\texorpdfstring{From geometry to generating functions:\newline rectangulations and permutations}{From geometry to generating functions: rectangulations and permutations}}

\author[Andrei Asinowski and Cyril Banderier]{Andrei Asinowski\thanks{Supported by FWF --- Austrian Science Fund, Grant P32731.}\addressmark{1} 
\and Cyril Banderier\addressmark{2}}

\address{\addressmark{1}{Institut für Mathematik, Universität Klagenfurt, Austria. {\normalfont\footnotesize\protect\textcolor{darkgreen}{\protect\url{https://me.aau.at/~anasinowski}}}}\newline 
\addressmark{2}{LIPN, CNRS \& Université~Sorbonne Paris Nord, France.  {\normalfont\footnotesize\protect\textcolor{darkgreen}{\protect\url{https://lipn.fr/~banderier}}}}}

\received{\today}

\abstract{We enumerate several classes of pattern-avoiding rectangulations. 
We establish new bijective links with pattern-avoiding permutations, prove that their generating functions are algebraic, and confirm several conjectures by Merino and Mütze.
We also analyse a new class of rectangulations, called whirls, using a generating tree.}


\begin{document}
\maketitle
\section{Introduction}
A \textit{rectangulation} of size $n$ is a tiling of a rectangle by $n$ rectangles 
such that no four rectangles meet in a point.  
In the literature, rectangulations are also 
called \textit{floorplans} or \textit{rectangular dissections}.
See Section~\ref{sec:weak_guil_sum} and~\cite{Asinowski24,Cardinal18,MerinoMuetze2023} for basic definitions and results.
  
Such structures appear naturally for architectural building plans,
integrated circuits (see Figure~\ref{fig:art}), 
and were investigated since the 70s with some graph theory, 
 computational\linebreak geometry, and combinatorial optimization point of views~\cite{Mitchell1976,Steadman1983}.
Then, in the 2000s, rectangulations began to be investigated with more 
combinatorial  approaches~\cite{Ackerman2006,AsinowskiBarequetBousquetMelouMansourPinter2013,AsinowskiMansour2010,Eppstein12,Reading12}:\linebreak
it was shown that some important families of rectangulations are enumerated by famous integer sequences (e.g., Baxter, Schröder, Catalan numbers) and that they have strong links with pattern-avoiding permutations (as studied in the seminal~article~\cite{ChungGraham1978}).

\bgroup \def\myheight{2.762cm}  \setlength\topsep{2pt}\setlength\partopsep{2pt}
\begin{figure}[bh]\centering \internallinenumbers
\includegraphics[height=\myheight]{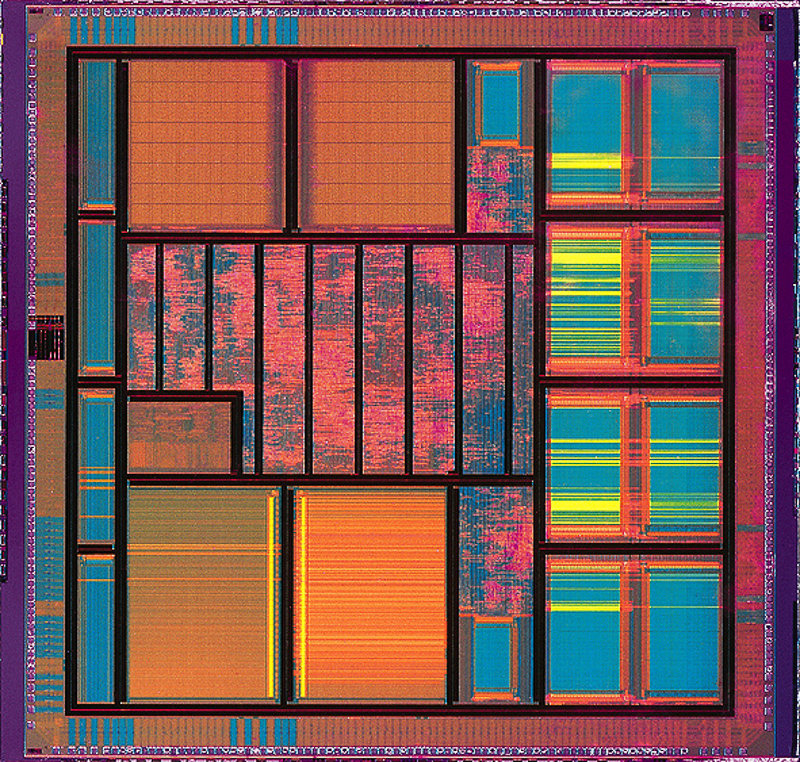} 
\qquad \qquad \includegraphics[height=\myheight]{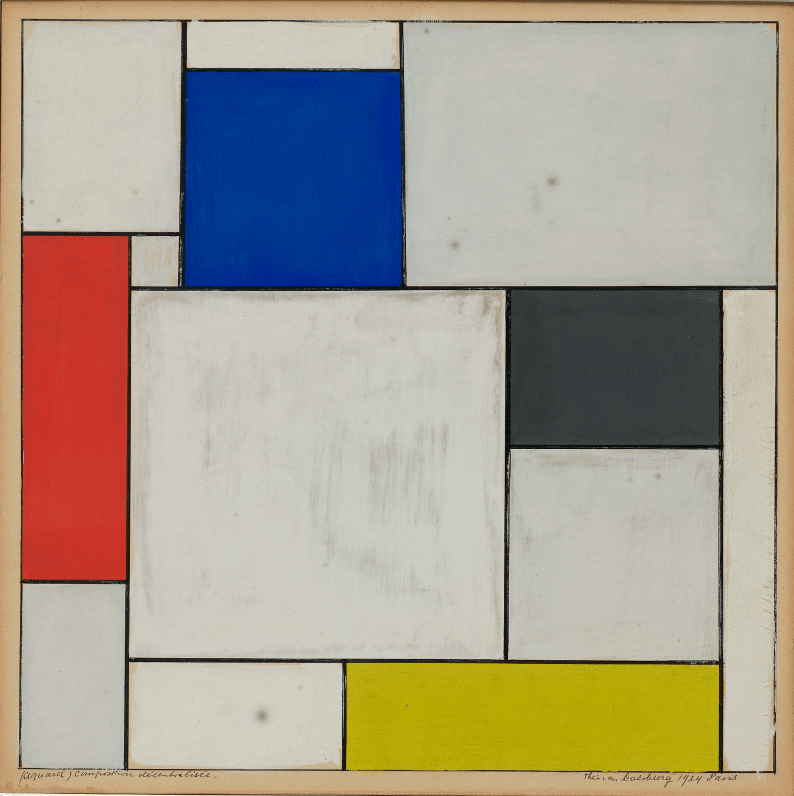}
\qquad \qquad \includegraphics[height=\myheight]{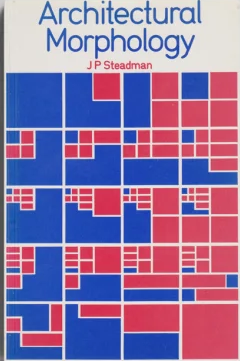}
\qquad \qquad \includegraphics[height=\myheight]{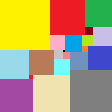}
\abovecaptionskip=-4mm
\belowcaptionskip=-6mm
\caption{{\protect\small  (a)
\href{https://en.wikipedia.org/wiki/Very_Large_Scale_Integration}{VLSI} are rectangulations playing an important role for 
 integrated circuits.\\
(b)  The artwork \textit{Composition d\'ecentralis\'ee}, 1924, by \href{https://en.wikipedia.org/wiki/Theo_van_Doesburg}{Theo van Doesburg} (1883--1931).\\
(c)  A book on the geometry of building plans~\cite{Steadman1983}. Its cover is not a rectangulation,  
since it contains instances of 4 rectangles meeting in a point.\\
(d) The minimal solution of Tutte's \href{https://tms.soc.srcf.net/about-the-tms/the-squared-square/}{"Squaring the square"} 
is a rectangulation~\cite{Duijvestijn78}. 
}}\label{fig:art}
\end{figure}
\egroup

\newpage
\section{Patterns in rectangulations and summary of our results}\label{sec:weak_guil_sum}

\begin{minipage}[b]{0.82\textwidth} \internallinenumbers
Two rectangulations are \textit{equivalent}
(``strongly equivalent'' in~\cite{Asinowski24})
if one can translate (horizontally or vertically) 
some of their segments (without meeting an endpoint of any other segment) so that they coincide. 
In the drawings on the right, only the first two rectangulations are equivalent.

\setlength{\parindent}{15pt}
In this article we deal with 
\textit{patterns in rectangulations}. 
In each drawing, we highlight an occurrence of the pattern \P2 (in green), \P3 (in blue),  \P7 (in red). 
A rectangulation contains \P7 if there is a (possibly further partitioned) rectangle
(here in gray)
such that the segment containing its left side has  an adjacent horizontal segment on the left,  
and the segment containing its  right side has an adjacent horizontal segment on its right.
\end{minipage}
\hspace{1pt}
\begin{minipage}[b]{0.15\textwidth} 
 
\includegraphics[scale=0.245]{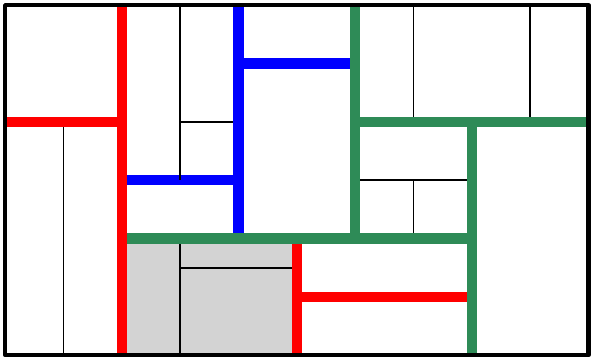}

\vspace{5pt}
\includegraphics[scale=0.245]{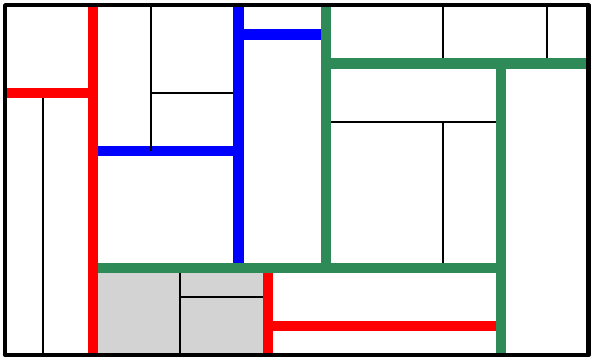}

\vspace{5pt}
\includegraphics[scale=0.245]{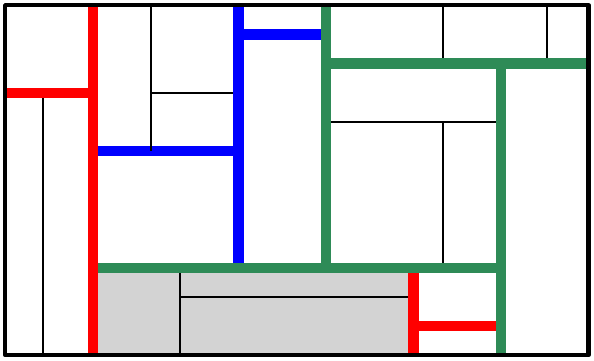}

\vspace{2pt}

\end{minipage}

We are interested in the enumeration of different natural classes of rectangulations,
where the goal is to count the number of non-equivalent rectangulations of size~$n$.\linebreak
E.g., \P3\P4-avoiding rectangulations are enumerated 
by Baxter numbers~\cite{Ackerman2006, ChungGraham1978}.

Recently, Arturo Merino and Torsten Mütze tackled the question of the
 exhaustive generation of rectangulations avoiding any subset of $\{\!\P1\!, \P2\!, \P3\!, \P4\!, \P5\!, \P6\!, \P7\!, \P8\!\}.$\linebreak
In~\cite{MerinoMuetze2023}, they present an efficient algorithm to generate such  
rectangulations\footnote{Let us here advertise the section dedicated to rectangulations in the nice \href{http://www.combos.org/rect}{Combinatorial Object Server},
created by Frank Ruskey, and now handled by Torsten Mütze, Joe Sawada, and Aaron Williams.}.
This led to a surprising observation: many sequences coincide (at least up to size $12$) 
with integer sequences which already appeared in the literature, for apparently unrelated problems.

In Theorem~\ref{thm:main}, 
we solve all the cases related to 
rectangulations avoiding \P1\P2\P3\P4.\linebreak
These are \textit{guillotine diagonal rectangulations}, 
that correspond to \textit{separable permutations}.
When they avoid further patterns
among \P5\P6\P7\P8, we obtain the following table\footnote{All other cases are equivalent to those presented here via straightforward symmetries.}, and\linebreak
provide generating functions for all these cases.
(See~\cite{Bona2022} for the notion of \textit{vincular~patterns}.)

\noindent
\begin{minipage}[b]{0.96\textwidth}
\begin{center}
\scalebox{0.903}{
\hspace{16pt}
\begin{tabular}{|c|c|c|c|c|}
\hline 
\begin{tabular}{c}
Entry in \\
\cite[Table 3]{MerinoMuetze2023}
\end{tabular}
&
\begin{tabular}{c}
Guillotine diagonal \\
rectangulations avoiding\ldots
\end{tabular}
& 
\begin{tabular}{c}
Separable permu-\\
tations avoiding\ldots
\end{tabular}
& G.f.  
& OEIS 
\\ \hline 
1234 &
$ \emptyset $     & $ 	\emptyset $ & alg.  
& \oeis{A006318}  \\  
12345 &
\P5   & $2\underline{14}3$  & alg.  
& \oeis{A106228}  \\  
12347 &
\P7    & $21354$ & alg.  
& \oeis{A363809}  \\  
123456 &
\P5\P6   & $2\underline{14}3, 3\underline{41}2$ & alg. 
& \oeis{A078482} \\  
123457 &
\P5\P7   & $2143$ & alg.
  
 & \oeis{A033321}   \\  
 123458 &
\P5\P8   & $2\underline{14}3, 45312$ &  alg.    & \oeis{A363810} \\  
123478 &
\P7\P8   & $21354, 45312$ &   rat.   & \oeis{A363811}   \\  
1234567 &
\P5\P6\P7  & $2143, 3\underline{41}2$ & alg.  &  \oeis{A363812}    \\  
1234578 &
\P5\P7\P8  & $2143, 45312$ & rat.   & \oeis{A363813}   \\  
12345678 &
\P5\P6\P7\P8 & $2143, 3412$ & rat.  
& \oeis{A006012} \\ \hline
\end{tabular}}
\end{center}
\end{minipage}
\smallskip

In Section~\ref{sec:vortex}, we additionally prove algebraicity of some non-guillotine models, such as 
\textit{vortex rectangulations}~(\oeis{A026029}, case 1345678 in~\cite{MerinoMuetze2023})
and \textit{whirls}~(\oeis{A002057}).

\newpage
\section{Guillotine diagonal rectangulations}\label{sec:guil}
The patterns $P_1=\P1$, $P_2=\P2$, $P_3=\P3$, $P_4=\P4$ were considered in some earlier 
work (for example~\cite{Ackerman2006, Cardinal18}) since they characterize some special kinds of rectangulations.

\smallskip
\noindent
\begin{minipage}[b]{0.74\textwidth}        \internallinenumbers
\ \ \ \   A rectangulation $\RR$ is \textit{guillotine} if it is of size $1$,
or if it has a \textit{cut} (a segment whose endpoints lie on opposite sides of~$R$) 
that splits it into two guillotine rectangulations.
It is well known~\cite{Ackerman2006} that a rectangulation is guillotine if and only if it avoids 
$P_1=\P1$ and $P_2=\P2$ (these two patterns are called \textit{windmills}). 
\end{minipage} 
\begin{minipage}[b]{0.25\textwidth}
\begin{center}
\ \includegraphics[scale=0.53]{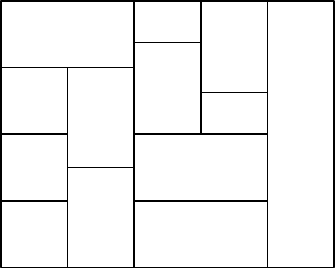} 
\end{center}
\end{minipage}

\noindent
\begin{minipage}[b]{0.79\textwidth}  \internallinenumbers
\ \ \ \   A rectangulation $\RR$ is called \textit{diagonal} if it avoids $P_3=\pthr$ and $P_4=\pfou$.
This notion is due to the fact that such a rectangulation can be drawn
so that the NW--SE diagonal of $R$ intersects all the rectangles.
At the same time, diagonal rectangulations are frequently seen as canonical representatives of 
rectangulations up to the ``weak equivalence''~\cite{Asinowski24, Cardinal18}.
\end{minipage}
\begin{minipage}[b]{0.2\textwidth}
\begin{center}
\ \includegraphics[scale=0.53]{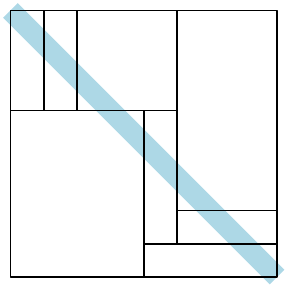} 
\end{center}
\end{minipage}

These two classes have --- in different ways --- stronger structural properties than the general case.
Therefore we expected that families 
that avoid these four patterns and any other subset of patterns  
of $\{\P5, \P6, \P7, \P8\}$
will yield noteworthy results. 
There are essentially ten different such models, all listed in~\cite[Table 3]{MerinoMuetze2023}. 
Amongst these 10 cases, 3 of them can be solved by ad-hoc bijections with trees (see~\cite{AsinowskiBarequetBousquetMelouMansourPinter2013,AsinowskiMansour2010}),
2 are conjectured by Merino and Mütze to lead to algebraic generating functions, 
and for the remaining 5 no conjectures were provided.
Below we present a unified framework which allows us to solve these 10 cases (confirming en passant the conjectures of Merino and Mütze).
The main result of this section is Theorem~\ref{thm:main}, 
which, in particular, states that all these cases are in fact algebraic!

\begin{thm}\label{thm:main}
The generating functions for the ten guillotine cases are algebraic.
\bgroup
\setlength{\belowdisplayskip}{5pt}
\begin{enumerate}
\item The generating function of  rectangulations avoiding \P1\P2\P3\P4 is
\begin{equation*}F(t)=\frac{1-t-\sqrt{1-6t+t^2}}{2}.\end{equation*}
\item The generating function of  rectangulations avoiding \P1\P2\P3\P4\P5 satisfies
\begin{equation*}tF^3 + 2tF^2 + (2t - 1)F + t=0.\end{equation*}
\item The generating function  of rectangulations avoiding \P1\P2\P3\P4\P7) satisfies\\
\bgroup \thinmuskip=0mu \medmuskip=0mu \thickmuskip=0mu
{\small $ t^4(t - 2)^2F^4 + 
 t(t - 2)(4t^3 - 7t^2 + 6t - 1)F^3 + (2t^4 - t^3 - 2t^2 + 5t - 1)F^2 - (4t^3 - 7t^2 + 6t - 1)F + t^2 = 0$.}
\egroup
\item The generating function of rectangulations avoiding \P1\P2\P3\P4\P5\P6 is
\begin{equation*}F(t)=\frac{1-3t+t^2-\sqrt{1-6t+7t^2-2t^3+t^4}}{2t}.\end{equation*}
\item The generating function of  rectangulations avoiding \P1\P2\P3\P4\P5\P7 is
\begin{equation*}F(t) = \frac{(1-t)(1-2t) - \sqrt{(1-t)(1-5t)}}{2t(2-t)}.\end{equation*}
\item The generating function of  rectangulations avoiding \P1\P2\P3\P4\P5\P8 satisfies

\bgroup
\thinmuskip=0mu \medmuskip=0mu \thickmuskip=0mu
\scalebox{0.981}
{\small $t^8(t - 2)^2F^4-t^3(t^2-3t+2) (t^5 - 7t^4 + 4t^3 - 6t^2 + 5t - 1)F^3 - t(t - 1)(4t^7 - 22t^6 + 37t^5 - 42t^4 + 53t^3 - 35t^2 +$}  \\
\scalebox{0.981}
{\small $10t - 1)F^2 $
$-(5t^6 - 16t^5 + 15t^4 - 28t^3 + 23t^2 - 8t + 1)(t - 1)^2F - (2t^5 - 5t^4 + 4t^3 - 10t^2 + 6t - 1)(t - 1)^2 = 0$.}
\egroup
\item The generating function of  rectangulations avoiding \P1\P2\P3\P4\P7\P8 is
\begin{equation*}F(t)=
\frac{t(1-16t+11t^2-434t^3+1045t^4-1590t^5+1508t^6-846t^7+252t^8-30t^9)}
{(1-2t)^4  (1-3t+t^2)^2  (1-4t+2t^2)}.\end{equation*}
\item The generating function of rectangulations avoiding \P1\P2\P3\P4\P5\P6\P7 is 
\begin{equation*}F(t) = \frac{1-3t-t^2 + 2t^3 -\sqrt{1-6t+7t^2+2t^3+t^4}}{2t^2(2-t)}.\end{equation*}
\item The generating function of  rectangulations avoiding \P1\P2\P3\P4\P5\P7\P8 is
\begin{equation*}F(t)=\frac{t(1-t)(1-7t+16t^2-11t^3+2t^4)}{(1-4t+2t^2)(1-3t+t^2)^2}.\end{equation*}
\item The generating function of rectangulations avoiding \P1\P2\P3\P4\P5\P6\P7\P8 is
\begin{equation*}F(t)=\frac{t(1-2t)}{1 - 4t + 2t^2}.\end{equation*}
\end{enumerate}
\egroup
\end{thm}

\smallskip

We now present \textit{separable permutations} --- a fundamental class which will be used in the proof of this theorem. This notion was coined in~\cite{bose98}.

\subsection{Separable permutations and rectangulations avoiding \P1\P2\P3\P4} \label{sec:1234}  
\medskip

A permutation $\pi$ is \textit{separable} if it is either of size $1$ (a \textit{singleton}),
or if it is (recursively) a direct sum of separable permutations 
(in this case $\pi$ is called \textit{ascending separable})
or a skew sum of separable permutations
(in this case $\pi$ is called \textit{descending separable}).
We refer to~\cite{Bona2022} for these notions.
Accordingly, separable permutations are precisely the non-empty
$(2413, 3142)$-avoiding permutations~\cite{bose98}.

The first key step in the proof of Theorem~\ref{thm:main} 
is ``translating'' (sets of) geometric patterns into (sets of) permutation patterns. 
In all 10 cases we obtain a bijection between a subclass of guillotine rectangulations
and a subclass of separable permutations. 
We provide details for the first three cases, and just give the key decompositions for the other cases.

\medskip

\medskip

\noindent\textbf{Case $1$: Guillotine diagonal rectangulations}.\label{lab:case1}
They are in bijection with separable permutations 
(see, e.g., \cite{Ackerman2006, AsinowskiMansour2010}).
Here is a natural recursive bijection:
the rectangulation of size 1 is mapped to the permutation of size 1,
and the recursive steps are illustrated in the following drawing.
The left and the middle illustrations describe 
the transformation for horizontal and vertical cuts,
and the right illustration is an example of size 11.

\medskip

\noindent\includegraphics[scale=0.71]{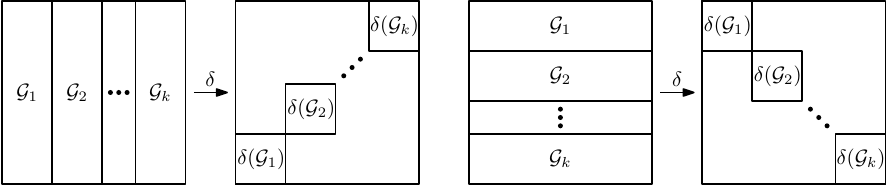} \hspace{4mm} \includegraphics[scale=0.71]{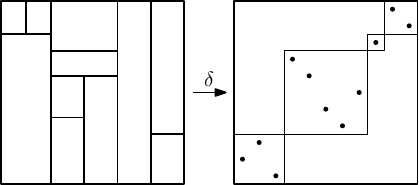} 

\medskip

The recursive definition of separable permutations translates directly
to a system of equations that binds $A(t)$, $D(t)$, and $F(t)=t+A(t)+D(t)$, the 
generating functions\linebreak for ascending, descending, and all separable permutations.
Since an ascending (resp.~descending) separable permutation can be seen as a sequence of singletons and descending (resp.~ascending) separable permutations (``blocks''),
we obtain the system  
$\left\{ {A = \frac{(t+D)^2}{1-(t+D)}}, \ \ 
{D = \frac{(t+A)^2}{1-(t+A)}} \right\}$.
Due to the symmetry $A(t)=D(t)$, we have
$A = \frac{(t+A)^2}{1-(t+A)}$.
This yields
$F(t)=\frac{1-t-\sqrt{1-6t+t^2}}{2}$,
the generating function of Schröder numbers (\oeis{A006318}).

\bigskip

\noindent\textbf{Case~2: $\P5$-avoiding guillotine diagonal rectangulations}.\label{lab:case2}

\begin{lem}~\label{thm:PP12345} 
A guillotine diagonal rectangulation $\RR$ avoids $\P5$
if and only if $\delta(\RR)$ avoids $2\underline{14}3$.
\end{lem}
\noindent 
\begin{minipage}[b]{0.55\textwidth}  \internallinenumbers
\textit{Proof (sketch).} 
This result follows from the bijection~$\delta$ described above. 
An occurrence of $\P5$ in $\RR$ means that there are four rectangles
$a, b, c, d$ as in the drawing, where the segment that separates $a$ and $b$ 
from $c$ and $d$ is a  cut at some step of the recursive decomposition of $\RR$.
It follows that in  
\end{minipage} 
\begin{minipage}[b]{0.4\textwidth}
\ \ \ \ \includegraphics[scale=0.86]{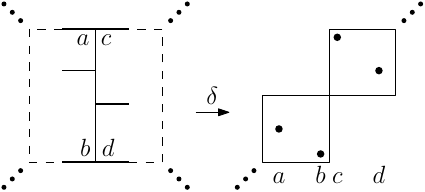} 
\end{minipage} \\
$\delta(\RR)$
we have four indices $a<b<c<d$ 
such that $a$ and $b$ belong to a descending block,
and~$c$ and~$d$ belong to the next descending block.
This yields an occurrence of $2\underline{14}3$ in $\delta(R)$.
The converse direction is based on similar considerations. \hfill $\qed$

\bigskip

\noindent\begin{minipage}[b]{0.73\textwidth}   \internallinenumbers
Now we enumerate $2\underline{14}3$-avoiding separable permutations.
Let $\pi$ be such a permutation.
If $\pi$ is ascending, then it either consists of at least two singletons,
or it has one or several descending blocks, which are separated by at least one 
singleton (see the drawing). For descending permutations, the decomposition is identical to Case~1,
since the skew sum of $2\underline{14}3$-avoiding ascending blocks cannot create a new occurrence 
of $2\underline{14}3$. This leads to the system
\end{minipage}
\begin{minipage}[b]{0.26\textwidth}
\begin{center}
\raisebox{0pt}{\includegraphics[scale=0.9]{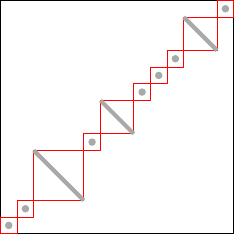}}
\end{center}
\end{minipage}
$\left\{A = \frac{t^2}{1-t} 
+ \left(\frac{1}{(1-t)^2}  \frac{1}{1- \frac{t D}{1-t}}  - 1\right)D, \ \ 
D = \frac{(t+A)^2}{1-(t+A)}\right\}$.
Solving this system (for example by computer algebra) yields Theorem~\ref{thm:main}(2).

\medskip
\noindent\textbf{Case~3: $\P7$-avoiding guillotine diagonal rectangulations}.\label{lab:case3}

First we show that
a guillotine diagonal rectangulation $\RR$ avoids $\P7$
if and only if $\delta(\RR)$ avoids $21354$.
As in Lemma~\ref{thm:PP12345}, 
the result directly follows from the definition of $\delta$.

Thus, we need to enumerate $21354$-avoiding separable permutations.
Let $\pi$ be an ascending separable permutation.
If $\pi$ has just one descending block, 
then $\pi$ avoids $21354$ if and only if this block avoids $21354$.
If $\pi$ has at least three descending blocks,
then it contains $21354$.
If $\pi$ has precisely two descending blocks, 
then $\pi$ is $21354$-avoiding if and only if they are adjacent, 
the first one is $213$-avoiding, and the second is $132$-avoiding.
An ascending $213$-avoiding permutation is 
either the identity permutation of size $\geq 2$,
or has at least one singleton and precisely one $213$-avoiding descending block (the last one). 
For descending permutations, 
an occurrence of $21354$ implies its occurrence in one of its ascending blocks:
hence, the decomposition is again identical to Case~1.
Let $S_A$ and $S_D$ be generating functions for ascending and, respectively, descending $213$-avoiding 
(and, equivalently, $132$-avoiding) permutations.
Then, we have 
$ 
\left\{S_A = \frac{t^2}{1-t}+ \frac{t S_D }{1-t}  , \ \
S_D = \frac{(t+S_A)^2}{1-(t+S_A)}\right\}
$, and  
for $21354$-avoiding separable permutations 
$\left\{{A = \frac{t^2}{1-t}+ \left(\frac{1}{(1-t)^2}-1\right) D  + \frac{S_D^2}{(1-t)^2} }, \ \
{D = \frac{(t+A)^2}{1-(t+A)}}\right\}$. 
These systems yield Theorem~\ref{thm:main}(3).

\vspace{13.3pt}

The treatment of other cases in Theorem~\ref{thm:main} is similar.
We first translate geometric patterns into permutation patterns,
obtaining some subclass of separable permutations.
Then its combinatorial specification yields a system of equations
that binds $A(t)$, the generating functions for ascending permutations in this class,
and $D(t)$ for descending permutations. In some cases we use an auxiliary family (as in Case~3 above).
Here we omit the details and only list permutation patterns, systems that bind $A(t)$ and $D(t)$, 
and, when relevant, auxiliary families and systems for their generating functions
$S_A$ and~$S_D$.

\medskip

\label{lab:case4}
\noindent\textbf{Case~4: $\{\P5, \P6\}$-avoiding guillotine diagonal rectangulations.} 
Such rectangulations are called \textit{one-sided guillotine rectangulations}~\cite{Felsner21}.
This family corresponds to $(2\underline{14}3, 3\underline{41}2)$-avoiding separable permutations.
Due to the symmetry of the model, we have $A(t)=D(t)$, and, therefore, just \textit{one} equation:
${A = \frac{t^2}{1-t} + \left(\frac{1}{(1-t)^2}  \frac{1}{1- \frac{tA}{1-t}} -1\right)A}$.

\medskip

\label{lab:case5}
\noindent\scalebox{0.95}[1]{\textbf{Case~5: $\{\P5, \P7\}$-avoiding guillotine diagonal rectangulations.}
They correspond to $2143$-} \\
avoiding separable permutations, the system is
$\left\{{A = \frac{t^2}{1-t} +  \left( \frac{1}{(1-t)^2}  -1 \right) D}, \ \ 
{D = \frac{(t+A)^2}{1-(t+A)}}\right\}$.

\medskip

\label{lab:case6}
\noindent\textbf{Case~6: $\{\P5, \P8\}$-avoiding guillotine diagonal rectangulations.} 
They correspond to $(2\underline{14}3, 45312)$-avoiding separable permutations.
The auxiliary class is $(2\underline{14}3, 231)$-avoiding permutations.
The system for the auxiliary class is
$\left\{{S_A = \frac{t^2}{1-t}+  \left( \frac{1}{1-\frac{tS_D}{1-t}} \frac{1}{(1-t)^2} - 1 \right) S_D }, \right.$
$\left.{S_D = \frac{t^2}{1-t}+ \frac{t S_A}{1-t} }\right\}$.
\  The system for $(2\underline{14}3, 45312)$-avoiding separable permutations is
$\left\{{A = \frac{t^2}{1-t}+ \left(\frac{1}{(1-t)^2}  \frac{1}{1-\frac{t D}{1-t}} -1\right) D }, \ \
{D = \frac{t^2}{1-t}+ \left( \frac{1}{(1-t)^2}-1\right)  A  + \frac{S_A^2 }{(1-t)^2}}\right\}$.

\medskip

\label{lab:case7}
\noindent\textbf{Case~7: $\{\P4, \P7, \P8\}$-avoiding guillotine diagonal rectangulations.} 
They correspond to $(21354, 45312)$-avoiding separable permutations.
The auxiliary class is $(45312, 213)$-avoiding permutations.
The system for the auxiliary class is
$\left\{{S_A = \frac{t^2}{1-t}+ \frac{S_D}{1-t}  }, \ \ 
S_D = \frac{t^2}{1-t}+ \right.$
$\left. \left(\frac{1}{(1-t)^2}-1\right)  S_A  
+ \left(\frac{1}{1-t}   \frac{t^2}{1-2t}\right)^2\right\}$.
The equation for $(21354, 45312)$-avoiding separable permutations is
${A =  \frac{t^2}{1-t} +  \left( \frac{1}{(1-t)^2}   - 1\right) A + \frac{S_D^2}{(1-t)^2} }$.

\medskip

\noindent\textbf{Case~8: $\{\P5, \P6, \P7\}$-avoiding guillotine diagonal rectangulations.} \label{lab:case8}
They correspond to $(2143, 3\underline{41}2)$-avoiding separable permutations.
This leads to the following system
$\left\{{A = \frac{t^2}{1-t}+ \left( \frac{1}{(1-t)^2}  -1 \right) D}, \right.$
$\left.{D = \frac{t^2}{1-t}+ \left( \frac{1}{(1-t)^2}  \frac{1}{1- \frac{t A}{1-t}}   - 1 \right) A}\right\}$.

\medskip

\label{lab:case9}
\noindent\textbf{Case~9: $\{\P5, \P7, \P8\}$-avoiding guillotine diagonal rectangulations.} 
They correspond to $(2143, 45312)$-avoiding separable permutations.
The auxiliary class is $(2143, 231)$-avoiding permutations.
The system for the auxiliary class is
$\left\{{S_A = \frac{t^2}{1-t}+ \left( \frac{1}{(1-t)^2}   - 1 \right) S_D }, \ \ 
S_D = \right.$ $\left. \frac{t^2}{1-t}+ \frac{t  S_A}{1-t}  \right\}$.
The system for $(2143, 45312)$-avoiding separable permutations is
$\left\{ A = \frac{t^2}{1-t}+ \right. $
$\left.\left(\frac{1}{(1-t)^2}  - 1 \right) D, \ \ 
{D = \frac{t^2}{1-t}+ \left(\frac{1}{(1-t)^2}-1 \right)A + \frac{S_A^2}{(1-t)^2}  }\right\}$.

\medskip
\label{lab:case10}
\noindent\textbf{Case~10: $\{\P5, \P6, \P7, \P8\}$-avoiding guillotine diagonal rectangulations.} 
They correspond to $(2143,3412)$-avoiding separable permutations.
The equation for this symmetric model is
${A = \frac{t^2}{1-t}+ \left( \frac{1}{(1-t)^2}  - 1 \right) A}$.

\label{lab:nonguil}
\section{Vortex rectangulations and whirls}\label{sec:vortex}
In this section we consider a class harder to enumerate, as it is not a guillotine case: rectangulations that avoid 
$\{\P1\P3\P4\P5\P6\P7\P8\}$ (that is, we forbid all our patterns except $P_2 = \P2$).
We denote this class of rectangulations by $\V$, and call them \textit{vortex} rectangulations.
Our ultimate goal is to prove the following conjecture by Merino and Mütze~\cite{MerinoMuetze2023}.

\begin{thm}\label{thm:1345678}
The generating function of $\V$ is $V(t)=tC^2(t)\big(1+t^2C^4(t)\big)$, where 
$C(t)=\frac{1-\sqrt{1-4t}}{2t}$ is the generating function of Catalan numbers.
The enumerating sequence of $\V$ is \oeis{A026029}.
\end{thm}

\label{def:whirl}
A vortex either avoids or contains the pattern $P_2 = \P2$. 
Vortices that avoid $\P2$ constitute Case~10 from Theorem~\ref{thm:main}. 
It remains to enumerate vortices with at least one $\P2$: 
such rectangulations will be called \textit{whirls}.
The \textit{interior} of a windmill is the (possibly further partitioned) rectangular area 
bounded by its segments. 
The entire rectangle being partitioned by a given rectangulation 
will be denoted by $R$.
  
\begin{lem}\label{thm:whirl_str}
If a whirl contains several windmills, then they are all nested.
In other words: for any two windmills, one of them entirely lies in the interior of the other.
\end{lem}

\begin{proof}[Proof (sketch).]
Let $W$ be a whirl, and consider some specific occurrence of $\P2$.
Starting from the right vertical segment of this windmill, 
we alternately go along the segments 
downwards to their lower endpoint 
and rightwards to their right endpoint,
until we reach the SE corner of $R$.
Similarly we define four \textit{alternating paths}: see Figure~\ref{fig:whirls} where they are shown by red.

These alternating paths partition $R$ into five regions:
$R_1$, $R_2$, $R_3$, $R_4$, and the interior of the windmill. 
In our drawings we colour these regions by blue, red, yellow, green, and grey.
Then every rectangle in $R_1$ and in $R_3$ has its top and bottom sides on the alternating paths,
and every rectangle in $R_2$ and in $R_4$ has its left and right sides on the alternating paths
(see Figure~\ref{fig:whirls}).
To prove this, for example for $R_1$, one scans this region from the left to the right:
then the assumption that some rectangle in $R_1$ violates this condition leads to an occurrence of 
$\pthr$, $\pfiv$ or $\psev$. 
Moreover, for every rectangle in $R_1$
its NW corner has the shape $\td$
and its SW corner has the shape $\tu$.
It follows that if another windmill --- not in the interior of the given one --- exists, then 
its segments belong to four different regions $R_1$, $R_2$, $R_3$, $R_4$.
Hence, the given windmill is entirely included in this another one.
\end{proof}

A whirl with an empty interior can be drawn 
so
that all the rectangles in $R_1$ and $R_3$ have width~$1$,
and 
all the rectangles in $R_2$ and $R_4$ have height $1$, and such a representation is unique.
To see that, we modify the whirl so that its segments belong
to consecutive vertical and horizontal grid lines.
See Figure~\ref{fig:whirls}(a) for an example of a whirl
which has two nested windmills (the corners of their interiors are shown by small dots). 

\begin{figure}
\begin{center}
\includegraphics[scale=0.7]{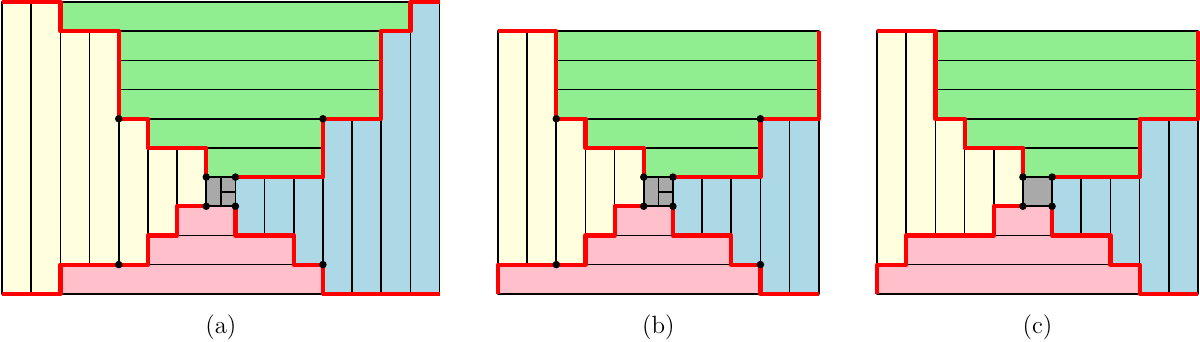} 
\end{center}
\abovecaptionskip=0pt
\caption{Three whirls: (a) is peelable, (b) is non-peelable, (c) is simple.}
\label{fig:whirls}
\end{figure}

A whirl is \textit{peelable} if it has a rectangle that extends from the 
top to the bottom or from the left to the right side of $R$.
From every peelable whirl it is possible to obtain a unique non-peelable whirl
by \textit{peeling} (i.e., successively deleting such rectangles).
Figure~\ref{fig:whirls}(b) shows a non-peelable whirl which is obtained from~\ref{fig:whirls}(a)
by peeling.

Finally, a \textit{simple} whirl is a non-peelable whirl with precisely one windmill $\P2$
whose interior is not further partitioned.
See Figure~\ref{fig:whirls}(c) for an example of a simple whirl.

\newpage
\subsection{Enumeration of simple whirls}
In this section we prove the following remarkable result: simple whirls are enumerated by $t^5C^4(t)$.
Our proof combines geometric-structural considerations, the generating tree method~\cite{BanderierBousquet-MelouDeniseFlajoletGardyGouyou-Beauchamps02}, and solving a functional equation with catalytic variables. 
It would be interesting to find an independent bijective proof. 

\subsubsection{Generating tree for simple whirls}\label{sec:gentree}
Denote by $S_1, S_2, S_3, S_4$ the eastern, southern, western, northern sides of~$R$.
A \textit{signature} of a simple whirl is the quadruple $(s_1,s_2, s_3, s_4)$,
where $s_i$ is the number of rectangles~from~$R_{i-1}$ 
that touch $S_i$. 
(The addition of indices of $R_i$'s and $S_i$'s is $\mathrm{mod} \ 4$.)
For example, the signature of the simple whirl from Figure~\ref{fig:whirls}(c)
is $(3, 2, 1, 2)$.

Given a simple whirl~$W$, it has exactly four \textit{corner rectangles} touching one~corner~of~$W$.
Since $W$ avoids \P1 and is not peelable, the corner rectangle touching the sides $S_i$ and $S_{i+1}$ has the same colour as the region $R_i$.
Then, as illustrated in the figure at the bottom of this page, from a simple whirl $W$ of size $n$, we can construct a simple whirl $W'$ of size $n+1$,
by adding a new corner rectangle to $R_i$ of length larger or equal to the length of the former corner rectangle
(and not touching the side $S_{i-1}$, to avoid creating a peelable whirl). 
Some rectangles of $R_{i-1}$ are then extended to reach the modified~$S_i$.
In $W'$, $s_{i+1}$~thus\linebreak increases by~$1$, and $s_{i}$ can assume all the values from $1$
to the (original)~$s_{i}$.
This generation algorithm has the drawback that some simple whirls are generated several times. 

To generate every simple whirl precisely once, we consider only those possibilities 
in which the added corner rectangle of $W'$ belongs to $R_i$ with the largest possible $i$
(that is, the largest $i$ such that in $W'$ we have $s_{i+1}>1$).
The new generation algorithm thus starts from the initial configuration $(1, 1, 1, 1)$  (only the unique simple whirl of size $5$ has this signature), and
applies the following rewriting rules
\begin{equation*}
\abovedisplayskip=0pt
\begin{array}{lcccl}
\text{1.} \ (a, \ b, \ 1, \ 1) \ \longrightarrow  \ (1, \ b+1, \ 1, \ 1), & & & &
\text{3.} \ (1, \ b, \ c, \ d) \ \longrightarrow  \ (1, \ b, \ [1..c], \ d+1),\\
\text{2.} \ (1, \ b, \ c, \ 1) \ \longrightarrow  \ (1, \ [1..b], \ c+1, \ 1),  & & & &
\text{4.} \ (a, \ b, \ c, \ d) \ \longrightarrow  \ (a+1, \ b, \ c, \ [1..d]).
\end{array}
\belowdisplayskip=3pt
\end{equation*}
The notation $[1..b]$ 
means that we generate $b$ signatures where this component takes the values $1, 2, \ldots, b$. 
The rules are not mutually exclusive: for example, all four rules can be applied on quadruples of the form $(1,b,1,1)$. 
The figure below shows all the descendants of a simple whirl $W$ with signature $(1, 2, 3, 1)$ on which the second, the third, and the fourth rules can be applied. 
The first rule does not apply since the resulting whirl $W'$ 
is obtained from a whirl different from $W$.
New corner rectangles are shown by bold boundary.\linebreak
\noindent\includegraphics[scale=0.753]{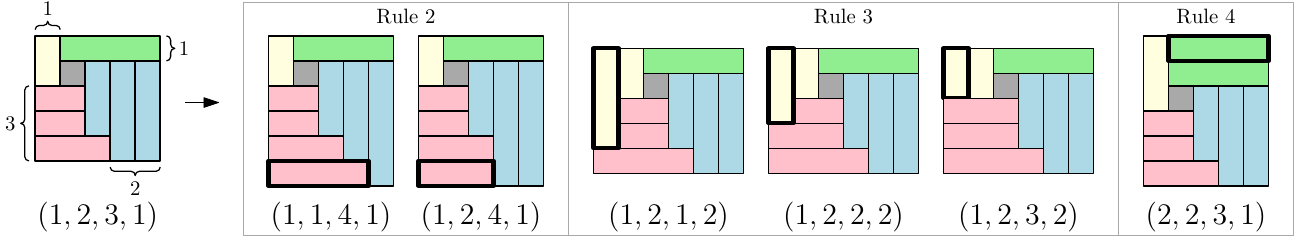} 

\newpage
\subsubsection{An intriguing functional equation}
\begin{thm}[Algebraicity of simple whirls]
Let $F(t,x_1,x_2,x_3,x_4)$ be the  multivariate generating function of simple whirls, 
where $z$ counts their size, 
and each $x_i$ counts the number of rectangles of colour $i$  touching their border. 
This generating function is algebraic and given by the  
closed form
\bgroup
\setlength{\abovedisplayskip}{4pt}
\setlength{\belowdisplayskip}{4pt}
\begin{equation}\label{Fclosedform}
F(t,x_1,x_2,x_3,x_4)= t^5 \frac {1}{2 \alpha} \left( \beta-\sqrt {\beta^2-4 \alpha e_4^2 } \right) 
\end{equation}
\egroup
with $\alpha:= \prod_{i=1}^4 (1-x_i+t x_i^2)$ and $\beta := (2 e_4 t^2 -  t (4e_4-3 e_3 +2e_2) + e_4-e_3+e_2 -e_1+2)e_4$, 
where $e_m:=  [t^m]  \prod_{i=1}^4  (1+tx_i)$  is the elementary symmetric polynomial of total degree~$m$.

In particular, the generating function of simple whirls is $F(t) = F(t,1,1,1,1)= t^5 C(t)^4$, where $C(t)$
is the generating function of Catalan numbers.
\end{thm}
\begin{proof}
The generating tree from Section~\ref{sec:gentree} translates to the  
functional equation
\bgroup
\setlength{\abovedisplayskip}{4pt}
\setlength{\belowdisplayskip}{4pt}
\begin{align}\label{FuncEq1}
\begin{split}
&F(t,x_1,x_2,x_3,x_4)=t^5x_1x_2x_3x_4+t x_1x_2 x_3x_4 [x_3 x_4]F (t,1,x_2,x_3,x_4) \\
&+tx_1 x_2 x_3 x_4 \frac{[x_1x_4]F(t,x_1,x_2,x_3,x_4)-[x_1x_4]F(t,x_1,1,x_3,x_4)}{x_2-1}\\
&+tx_1 x_3 x_4 \frac{[x_1] F(t,x_1,x_2,x_3,x_4)- [x_1] F(t,x_1,x_2,1,x_4)}{x_3-1}\\
&+t x_1 x_4 \frac{F(t,x_1,x_2,x_3,x_4)-F(t,x_1,x_2,x_3,1)}{x_4-1}.
\end{split}
\end{align}
\egroup
Unfortunately, there are currently no generic methods to solve this type of catalytic functional equation.
Luckily, in our case, we were able to solve this equation; 
we invite the readers to give it a try before going on with their reading! 

\bgroup \thinmuskip=0.9\thinmuskip \medmuskip=.9\medmuskip \thickmuskip=0.71\thickmuskip
Here is how we solved it. First, recall that 
the valuation of a  
series $f(t)=\sum_{n\geq 0} f_n t^n$ is\linebreak the smallest integer~$n$ such that $f_n\neq0$ 
(and $\operatorname{val}(f(t))=+\infty$ if $f(t)=0$).
Thus,  Equation~\eqref{FuncEq1}\linebreak is a contraction in the metric space of formal power series
(equipped with the  
distance $d(f(t),g(t)) = 2^{-\operatorname{val}(f(t)-g(t))}$).
Therefore, the Brouwer fixed-point theorem ensures that\linebreak there is a unique  series $F$
satisfying Equation~\eqref{FuncEq1}.
Now, it can be checked (by substitution) that the closed form~\eqref{Fclosedform} satisfies the functional equation~\eqref{FuncEq1};
this proves the theorem. 
\egroup

Let us also explain how we guessed this closed form, as it offers a useful heuristic for dealing with similar equations.
The classical guessing technique using Padé approximants is too costly, so, instead, we used 
linear algebra to identify an algebraic equation of degree 2 (in $F$) and degree~2  (in $x_1$) for $F(t,x_1,11,31,71)$.  
It is not obvious from the functional equation that $F(t,x_1,x_2,x_3,x_4)$ is a symmetric function in the $x_i$'s --- yet, this follows from the fact that any rotation of a whirl is still a whirl.
Therefore, its minimal polynomial should also have symmetric coefficients in the $x_i$'s.
Then, when one obtains a monomial like $ 532642 x_1 = 2 x_1 \times 11^2\times 31 \times 71$, 
it makes sense to rewrite it as $2 x_1 x_2^2 x_3 x_4$,  and all the symmetric versions of this monomial will also appear as coefficients.
This leads to the minimal polynomial $\alpha G^2-\beta G+ e_2^4$ (for $G=F/t^5$) , and thus to the closed form~\eqref{Fclosedform}. 
\end{proof}

\subsection{Enumeration of whirls and vortices}

We go back from simple whirls to possibly peelable whirls with empty interior
by alternately adding sequences of rectangles on the two horizontal and the two vertical sides.
This yields the generating function
$P(t)= t^5C^4(t)  \bigg(\frac{2}{1-\left(\frac{1}{(1-t)^2} -1\right) }-1\bigg).$  

Such whirls can be transformed into a whirl with $>1$ windmills by substituting the interior by another whirl (see Figure~\ref{fig:whirls}).
Therefore, whirls $\mathcal W$ with empty interior of the innermost windmill
are enumerated as a sequence of $P(t)/t$. 
Thus we obtain the generating function $W(t)= \frac{1}{1-P(t)/t}$.

Finally, to get the family $\mathcal V$ of vortices (i.e., the rectangulations that avoid \mbox{$\P1, \P3, \P4,$} $\P5,\P6,\P7,\P8$), 
we replace the innermost windmill by a rectangulation that avoids all eight patterns 
(i.e., Case 10 from Theorem~\ref{thm:main}, counted by $Z(t)=\frac{t(1-2t)}{1 - 4t + 2t^2}$).
This~leads~to
\bgroup
\begin{equation*}
V(t)=W(t) Z(t) = (1-2t)\big( \, 1-4t+2t^2+(1-2t)\sqrt{1-4t} \, \big) /(2t^3)
=tC^2(t)(1+t^2C^4(t)),
\end{equation*}
\egroup
 which is exactly the generating function of the sequence \oeis{A026029},
 as conjectured in~\cite[Table~3, entry 1345678]{MerinoMuetze2023}.
  This concludes the proof of Theorem~\ref{thm:1345678}.\smallskip

As for any algebraic generating function, 
the corresponding sequence satisfies a~linear recurrence,
\bgroup \thinmuskip=0mu \medmuskip=0mu \thickmuskip=0mu
$(n+4)v_n -6(n+2) v_{n-1}+4 (2n-1) v_{n-2}=0$,
\egroup
from which one can compute~$v_n$ in time $O(\sqrt n \ln n)$ (counting each arithmetical operation with cost 1),
and singularity analysis gives
$v_n\sim 4^{n+2}/\sqrt{\pi}  n^{-3/2}$.
(See~\cite{BanderierDrmota15,FlajoletSedgewick09} for more references on all these facts.)

\section{Conclusion}

In this article, we solved several conjectures related to 
families of pattern-avoiding rectangulations and permutations.
We proved that all our generating functions are $\mathbb N$-algebraic\footnote{This is the class of generating functions counting words of length $n$ generated by unambiguous context-free grammars. It has many noteworthy structural and asymptotic properties~\cite{BanderierDrmota15,BousquetMelou06,ChomskySchutzenberger63}.},
and we  provide an interesting  example of $\N$-algebraic structure 
(the simple whirls, counted by $t^5 C^4(t)$) for which no context-free grammar is known.

Merino and Mütze~\cite[Table 3]{MerinoMuetze2023} are mentioning a few more families of rectangulations
for which the enumeration is still open. 
Some of them are in fact tractable with variants of methods presented in our article. 
These results will be included in the long version of our article.
It would also be of interest to consider further forbidden patterns, e.g., to determine which patterns  lead to algebraic, D-finite, D-algebraic generating functions. 
Also, is it the case that they all lead to a Stanley--Wilf-like conjecture: is the number of such rectangulations bounded by $A^n$, for some constant $A$?

In conclusion, the rectangulations, while having a very simple definition, are an inexhaustible source of challenging problems for generating function lovers!

\newpage
\small
\newlength{\bibitemsep}\setlength{\bibitemsep}{.125\baselineskip plus .05\baselineskip minus .05\baselineskip}
\newlength{\bibparskip}\setlength{\bibparskip}{0pt}
\let\oldthebibliography\thebibliography
\renewcommand\thebibliography[1]{ 
  \oldthebibliography{#1} 
  \setlength{\parskip}{\bibitemsep} 
  \setlength{\itemsep}{\bibparskip} 
}
\bibliographystyle{SLC}
\bibliography{AB}
\end{document}